\documentclass[journal]{IEEEtran}
\usepackage{amsmath}
\usepackage{amssymb}
\usepackage{amsfonts}
\usepackage{graphicx}
\usepackage{enumerate,float,color,graphicx,fancybox,pifont,epsf,epsfig,subfigure,amsmath,amssymb,psfrag,amsthm}
\usepackage{pstool} 
\usepackage{cite}

\newtheorem{theorem}{\textbf{Theorem}}
\newtheorem{lemma}[theorem]{\textbf{Lemma}}
\newtheorem{proposition}[theorem]{\textbf{Proposition}}
\newtheorem{remark}[theorem]{\textbf{Remark}}

\newcommand{\secref}[1]{Section~\ref{#1}}

\newcommand{\figref}[1]{Figure~\ref{#1}}

\newcommand{\theoref}[1]{Theorem~\ref{#1}}

\newcommand{\proref}[1]{Proposition~\ref{#1}}
\newcommand{\lemref}[1]{Lemma~\ref{#1}}

\def\Tr{\mathrm{Tr}}

\title{Optimized Compressed Sensing Matrix Design for Noisy Communication Channels}
\author{Amirpasha Shirazinia and  Subhrakanti Dey \\
Signals \& Systems Division, Department of Engineering Sciences, Uppsala University, Sweden \\
\texttt{Email: amirpasha.shirazinia@signal.uu.se, subhrakanti.dey@signal.uu.se}}

\begin{document}
\maketitle

\begin{abstract}
We investigate a power-constrained sensing matrix design problem for a compressed sensing framework. We adopt a mean square error (MSE) performance criterion for sparse source reconstruction in a system where the source-to-sensor channel and the sensor-to-decoder communication channel are noisy. Our proposed sensing matrix design procedure relies upon minimizing a lower-bound on the MSE.  Under certain conditions, we derive closed-form solutions to the optimization problem. Through numerical experiments, by applying practical sparse reconstruction algorithms, we show the strength of the proposed scheme by comparing it with other relevant methods. We discuss the  computational complexity of our design method, and develop an equivalent stochastic optimization method to the problem of interest that can be solved approximately with a significantly less computational burden. We illustrate that the  low-complexity method still outperforms the popular competing methods.

\end{abstract}

\vspace{-0.25cm}
\section{Introduction} \label{sec:intro}
Compressed sensing (CS) \cite{06:Donoho,06:Candes,08:Candes} is an emerging tool for simultaneous signal acquisition and compression that significantly reduces the cost due to sampling, leading to low-power consumption and low-bandwidth communication. CS is indeed a mathematical framework, based on linear dimensionality reduction, and builds upon the fact that the source signal can be represented in a sparse form, which is true for many physically observed signals.

In order to clarify the concept of CS in relation to the objectives of our work, let us consider the linear reduction model $\mathbf{y = A x + n}$, where $\mathbf{x} \in \mathbb{R}^N$ is a sparse vector (in a known basis) with a size higher than that of the measurement vector $\mathbf{y} \in \mathbb{R}^{M}$. Further, $\mathbf{A} \in \mathbb{R}^{M \times N}$ is a \textit{fat} sensing matrix (i.e., $M < N$), and $\mathbf{n} \in \mathbb{R}^{M}$ is the measurement noise vector. It should be mentioned that a careful design of the sensing matrix $\mathbf{A}$ is crucial in order to achieve promising performance of sparse reconstruction algorithms. Moreover, as shown analytically in \cite{14:Shirazi},  the sensing matrix has an important role in not only determining the amount of estimation error, but also in characterizing the amount of distortion due to quantization and transmission of CS measurements over noisy  communication channels. Therefore, in this paper, we are interested in the optimized design of the sensing matrix $\mathbf{A}$ with respect to an appropriate performance criterion reflecting the mean square estimation error due to transmission over a noisy communication channel.

In the literature, available approaches for designing sensing matrices for estimation purposes can be divided into  three main kinds:

	1) In the first category, the  design method is linked to a fundamental feature of the sensing matrix $\mathbf{A}$, called mutual coherence \cite{01:Donoho}, which is defined as follows
\begin{equation} \label{eq:mutual co}
    \mu \triangleq  \underset{i \neq j}{\max} \hspace{0.2cm} \frac{|\mathbf{A}_i^\top \mathbf{A}_j|}{\|\mathbf{A}_i\|_2 \|\mathbf{A}_j\|_2}, \hspace{0.2cm} 1 \leq i,j \leq N,
\end{equation}
where $\mathbf{A}_i$ denotes the $i^{th}$ column of $\mathbf{A}$. One of the early works within this category is \cite{07:Elad} which studied algorithmic methods in order to minimize the mutual coherence.
	
	2) In the second category, in order to analytically address the sensing matrix design problem in a more  tractable manner, and to reduce the amount of mutual coherence, the sensing matrix $\mathbf{A}$ is optimized with respect to satisfying
	\begin{equation} \label{eq:2nd cat}
	\begin{aligned}
		&\underset{\mathbf{A}}{{\text{minimize}}} \hspace{0.25cm} \|  \mathbf{A}^\top \mathbf{A}  - \mathbf{I}_N \|_F ,& \\
	\end{aligned}
	\end{equation}
 where $\| \cdot \|_F$ denotes the Frobenius norm, and $(\cdot)^\top$ denotes the matrix transpose. Some works in this category are \cite{11:Zelnik,10:vahid,10:vahid2,13:Gang,09:Duarte}.
	
	3) While in the first and second categories, the sensing matrix is mainly designed in a way to address the worst-case performance of sparse reconstruction, the actual performance, such as mean square error (MSE) of sparse reconstruction, is typically far less. Thus, one might consider minimization of
\begin{equation} \label{eq:MSE cat3}
	\mathrm{MSE} \triangleq \mathbb{E}[\|\mathbf{x} - \widehat{\mathbf{x}}\|_2^2],
\end{equation}
under relevant constrains. Here, $\| \cdot \|_2$ denotes the $\ell_2$ norm, and $\widehat{\mathbf{x}}$ represents the output of decoder (e.g, a  linear or non-linear estimator, sparse reconstruction algorithms  etc.). Some examples within this category are \cite{12:Chen,08:Jin,07:Schizas,14:Yuan}.

Following the third category, we are interested in the optimized sensing matrix design with respect to minimizing reconstruction MSE criterion given that the source can be represented in a sparse form with known statistical moments.

We study a scenario that a correlated sparse source vector (i.e., the non-zero components of the source signal are correlated) is scaled linearly and becomes corrupted by additive noise before compression/encoding via a CS-based sensing matrix. The resulting CS measurements are transmitted over a noisy (analog) communication channel, modeled by channel gain and additive noise, under available average transmit power constraint. At the receiving-end, the source signal is decoded using an estimator in order to reconstruct the sparse source. In this scenario, we aim at optimizing the sensing matrix with respect to minimizing a \textit{lower-bound} on the MSE incurred by using the MMSE estimator (which by definition minimizes the MSE) of a sparse source signal. We propose a two-stage sensing matrix optimization procedure that combines semi-definite relaxation (SDR) programming as well as low-rank approximation problem. The solution to the low-rank approximation problem can be derived analytically, further, the SDR programming can be solved using convex optimization techniques. Also, under certain conditions, we derive closed-form solutions to the SDR problem. Through numerical experiments, by applying practical sparse reconstruction algorithms, we compare our proposed scheme with other relevant methods. Experimental results show that the proposed approach improves the MSE performance by a large margin compared to other methods. This performance improvement is achieved at the price of higher computational complexity. In order to tackle the complexity issue, we develop an equivalent stochastic optimization method to the problem of interest, which can be approximately solved, and still shows a superior performance over the competing methods.

\vspace{-0.2cm}
\section{System Description} \label{sec:problem}

\subsection{System Model and Key Assumptions} \label{sec:sys model}

We study the setup shown in \figref{fig:diagram}.
\begin{figure} [t]
  \centering
  \psfrag{E}[][][0.75]{CS encoder}
  \psfrag{C}[][][0.75]{Channel}
   \psfrag{D}[][][0.75]{Decoder}
  \psfrag{x}[][][0.9]{$\mathbf{x}$}
  \psfrag{A}[][][0.9]{$\mathbf{A}$}
  \psfrag{y}[][][0.8]{$\mathbf{y}$}
  \psfrag{z}[][][0.8]{$\mathbf{z}$}
   \psfrag{v}[][][0.9]{$\mathbf{v}$}
  \psfrag{w}[][][0.9]{$\mathbf{w}$}
  \psfrag{H}[][][0.9]{$\mathbf{H}$}
  \psfrag{G}[][][0.9]{$\mathbf{G}$}
  \psfrag{h}[][][0.9]{$\widehat{\mathbf{x}}$}
    \includegraphics[width=9cm]{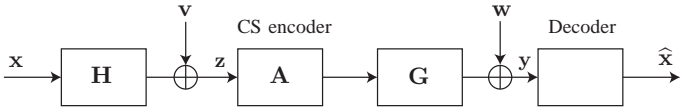}    
  \caption{System model for CS over a noisy communication channel.}\label{fig:diagram}
 \vspace{-0.5cm}
  \centering
\end{figure}
We consider a $K$-sparse (in a known basis) vector $\mathbf{x} \in \mathbb{R}^N$ which is comprised of exactly $K$ random non-zero components ($K \ll N$). We define the support set, i.e., the locations of the non-zero component, for the vector $\mathbf{x} \triangleq [x_1,\ldots,x_N]^\top$ as $\mathcal{S} \triangleq \{n \in \{1,2,\ldots,N\}: x_n \neq 0 \}$ with $|\mathcal{S}| = K$, where $|\cdot|$ denotes the cardinality of a set. We assume that the non-zero components of the source vector $\mathbf{x}$ are distributed according to a Gaussian distribution $\mathcal{N}(\mathbf{0}_K,\mathbf{R})$, where $\mathbf{R} = \mathbb{E}[\mathbf{x}_\mathcal{S} \mathbf{x}_\mathcal{S}^\top] \in \mathbb{R}^{K \times K}$ is the covariance matrix of the $K$ non-zero components of $\mathbf{x}$, and $\mathbf{x}_\mathcal{S} \in \mathbb{R}^K$ denotes the components of $\mathbf{x}$ indexed by the support set $\mathcal{S}$. Note that $\mathbf{R}$ is a positive definite matrix which is not necessarily scaled identity, i.e., the nonzero off-diagonal elements of $\mathbf{R}$ allow  the non-zero components of $\mathbf{x}$ to be correlated. The elements of the support set $\mathcal{S}$ are drawn uniformly at random from the set of all ${N \choose K}$ possibilities, denoted by $\Omega$, i.e., $|\Omega|= {N \choose K}$. In other words, $p(\mathcal{S}) = 1/{N \choose K}$, where $p(\mathcal{S})$ represents the probability that a support set $\mathcal{S}$ is chosen from the set $\Omega$. The uniform distribution is chosen for simplicity of presentation, however, extensions to other types of distributions  are straightforward. We also denote the covariance matrix of the whole sparse source vector by $\mathbf{R}_x \triangleq \mathbb{E}[\mathbf{xx}^\top] \in \mathbb{R}^{N \times N}$.

We model the uncertainty or mismatch in some physical aspect via a source-to-sensor channel described as following. The source is linearly scaled via a fixed matrix $\mathbf{H} \in \mathbb{R}^{L \times N}$ whose output is corrupted by an additive noise $\mathbf{v} \in \mathbb{R}^L$ uncorrelated with the source, where $\mathbf{v} \sim \mathcal{N}(\mathbf{0}_L,\sigma_v^2 \mathbf{I}_L)$. For transmission over noisy channel, the noisy observations should be compressed and then encoded. Here, we assume that the bandwidth of the noisy observation $\mathbf{z \triangleq Hx + v} \in \mathbb{R}^L$ is compressed via a full row-rank compressed sensing transformation matrix $\mathbf{A} \in \mathbb{R}^{M \times L}$, where $M < L$. We also assume that $M < N$. The compressed measurements are simultaneously encoded under the constraint of the available average transmit power, and then transmitted over a channel, represented by a fixed channel matrix $\mathbf{G} \in \mathbb{R}^{M \times M}$ and additive noise $\mathbf{w} \in \mathbb{R}^M$. We assume that the channel matrix is given by $\mathbf{G} = g \mathbf{I}_M$, and we let the additive channel noise be distributed as $\mathbf{w} \sim \mathcal{N}(\mathbf{0}_M,\sigma_w^2 \mathbf{I}_M)$, which is uncorrelated with the source $\mathbf{x}$ and source-to-sensor additive  noise $\mathbf{v}$. The rational behind the scaled identity assumption on the channel matrix is that there is no inter-symbol interference between message transmissions over the communication link, and the channel is assumed to remain constant during each observation period \cite{08:Jin}. This technical assumption also makes our design procedure tractable. In a more compact way, we write 
\begin{equation} \label{eq:measurement}
\begin{aligned}
	\mathbf{y} &= \mathbf{GAz + w} =   g\mathbf{AH x} + \underbrace{g\mathbf{Av + w}}_{\triangleq  \mathbf{n}}.& 
\end{aligned}
\end{equation}
Denoting the total noise in the system by $\mathbf{n} \triangleq g\mathbf{Av + w} \in \mathbb{R}^M$, its covariance matrix  $\mathbf{R}_n \in \mathbb{R}^{M \times M}$ can be calculated as 
\begin{equation} \label{eq:cov nosie}
\begin{aligned}
	\mathbf{R}_n \triangleq \mathbb{E}\{\mathbf{nn}^\top\} = g^2 \sigma_v^2 \mathbf{AA}^\top + \sigma_w^2 \mathbf{I}_M.
\end{aligned}
\end{equation}

Finally, at the receiving end, the decoder which is characterized by a (potentially non-linear) mapping $\mathbb{R}^M \!\rightarrow \! \mathbb{R}^N$ provides the estimate of the source from corrupted measurements. 

\vspace{-0.2cm}
\subsection{Developing MMSE Estimation and Lower-bound on MSE} \label{sec:perf}

We are interested in designing an optimized compressed sensing matrix $\mathbf{A}$ with respect to minimizing the MSE of sparse source reconstruction. Based on the aforementioned assumptions in \secref{sec:sys model}, it is possible (see e.g. \cite{09:Elad}) to find a closed-form expression for the MMSE estimation of the source given the received signal vector $\mathbf{y}$. The MMSE estimator, denoted by $\widehat{\mathbf{x}}^\star \!  \! \in \! \! \mathbb{R}^N$, minimizes the MSE by definition, and inherits the following structure (see e.g. \cite{09:Elad})
\begin{equation} \label{eq:struct MMSE}
	\widehat{\mathbf{x}}^\star = \sum_{\mathcal{S} \subset \Omega} \beta (\mathcal{S},\mathbf{y}) \mathbb{E}[\mathbf{x}| \mathbf{y},\mathcal{S}],
\end{equation}
where $\Omega$ is the set of all sparsity patterns, and $\beta(\mathcal{S},\mathbf{y})$ is the weighting coefficient (possibly non-linear in $\mathbf{y}$) such that $\sum_\mathcal{S} \beta(\mathcal{S},\mathbf{y}) \!=\! 1$. Further, $\mathbb{E}[\mathbf{x}| \mathbf{y},\mathcal{S}] $ is the conditional mean of the source given a possible support set $\mathcal{S}$ and observation $\mathbf{y}$. The conditional mean in \eqref{eq:struct MMSE}  can be expressed as (see  \cite{09:Elad}) $\mathbb{E}[\mathbf{x}| \mathbf{y},\mathcal{S}] = $
\begin{equation} \label{eq:oracle MMSE est}
	g \left(\mathbf{R}^{-1} \! + g^2   \left(\mathbf{H}^\top \mathbf{A}^\top \right)_\mathcal{S} \; \mathbf{R}_n^{-1} \; \left(\mathbf{A} \mathbf{H} \right)_\mathcal{S} \right)^{-1} \left(\mathbf{H}^\top \mathbf{A}^\top \right)_\mathcal{S} \; \mathbf{R}_n^{-1} \mathbf{y},
\end{equation}
where $(\cdot)_\mathcal{S}$ denotes the columns of a matrix indexed by the support set $\mathcal{S}$. The MMSE estimator \eqref{eq:struct MMSE} gives the lowest possible MSE for a sparse source in the system of \figref{fig:diagram}. However, the MSE itself  does not have a closed-form expression, which makes it difficult to find a tractable way in order to optimize the sensing matrix. Thus, we propose an alternative sensing matrix optimization method by minimizing a lower-bound on the MSE.


We bound the MSE of the MMSE estimator by that of the oracle MMSE estimator, i.e., an ideal estimator which has the perfect knowledge of the support set \textit{a priori}. By definition, the oracle estimator is calculated as the conditional expectation $\widehat{\mathbf{x}}^{(or)} \triangleq \mathbb{E}[\mathbf{x}| \mathbf{y},\mathcal{S}]$, as shown in \eqref{eq:oracle MMSE est}, given \textit{a priori} known (but random) support set $\mathcal{S}$ and noisy observations $\mathbf{y}$. Notice that the conditional expectation given the support set is Gaussian distributed which gives the following MSE
\begin{equation} \label{eq:oracle_MSE}
\begin{aligned}
	&\mathrm{MSE}^{(lb)} = \mathbb{E}[\| \mathbf{x} - \widehat{\mathbf{x}}^{(or)} \|_2^2] =  \mathbb{E}[\| \mathbf{x}_\mathcal{S} - \widehat{\mathbf{x}}_\mathcal{S}^{(or)} \|_2^2] &\\
	&\overset{(a)}{=} \sum_{\mathcal{S} \subset \Omega} p(\mathcal{S})  \Tr\left\{ \left(\mathbf{R}^{-1} + g^2 (\mathbf{H}^\top \mathbf{A}^\top)_\mathcal{S} \;\mathbf{R}_n^{-1} \; (\mathbf{A} \mathbf{H})_\mathcal{S} \right)^{-1} \right\},&
\end{aligned}
\end{equation}
where $(a)$ follows by averaging over all random supports sets. Further, $p(\mathcal{S})= 1 / {N \choose K}$.

To be able to formulate the MSE in \eqref{eq:oracle_MSE} in terms of the sensing matrix $\mathbf{A}$, as in \cite{12:Chen}, we define the matrix $\mathbf{E}_\mathcal{S} \in \mathbb{R}^{N \times K}$, which is formed by taking an  identity matrix of order $N \times N$ and deleting the columns indexed by the support set $\mathcal{S}$.
 Then, we rewrite 
\begin{equation} \label{eq:oracle_MSE_2}
	\mathrm{MSE}^{(lb)} \! \!= \! \sum_{\mathcal{S}} \! \frac{1}{{N \choose K}} \! \Tr \left\{ \left(\mathbf{R}^{-1} \!+\! g^2 \mathbf{E}_\mathcal{S}^\top \mathbf{H}^\top \mathbf{A}^\top \mathbf{R}_n^{\! -1} \mathbf{A} \mathbf{H} \mathbf{E}_\mathcal{S} \right)^{\!-1} \right\}.
\end{equation}

\vspace{-0.3cm}
\section{Design Methodology} \label{sec:design}
In this section, we offer a design method for optimization of the sensing matrix $\mathbf{A}$ with the objective of minimizing the lower-bound \eqref{eq:oracle_MSE_2}. The optimization is performed at the decoder, and we assume that the decoder knows the sensor observation models and source-to-sensor and sensor-to-decoder channels. 

We assume that the bandwidth is constrained, i.e., we have $M < N$ total number of observations. Also, 
the average transmit power can be bounded by the total available power $P$ as follows
\begin{equation} \label{eq:power}
\begin{aligned}
	\mathbb{E}[\|\mathbf{AHx + Av} \|_2^2] 
	&= \Tr \{\mathbf{AH}\mathbb{E}[\mathbf{xx}^\top]\mathbf{H}^\top \mathbf{A}^\top \! + \! \mathbf{A} \mathbb{E}[\mathbf{vv}^\top] \mathbf{A}^\top\}& \\
	&= \Tr \{\mathbf{AH} \mathbf{R}_x \mathbf{H}^\top \mathbf{A}^\top + \sigma_v^2 \mathbf{AA}^\top\} \leq P.&
\end{aligned}
\end{equation}

Minimizing the lower-bound \eqref{eq:oracle_MSE_2} subject to the average power constraint \eqref{eq:power} yields 
\begin{equation} \label{eq:opt 1}
\begin{aligned}
	&\underset{\mathbf{A}}{\text{minimize}} \hspace{0.25cm} \mathrm{MSE}^{(lb)}& \\
	& \text{subject to} \hspace{0.25cm} \Tr \{\mathbf{A} (\mathbf{HR}_x \mathbf{H}^\top + \sigma_v^2 \mathbf{I}_N) \mathbf{A}^\top\} \leq P .& 
\end{aligned}
\end{equation}
The optimal solution to  Problem \eqref{eq:opt 1} is equivalent to that of the optimization problem given by the following theorem.
  
\begin{theorem} \label{theo:sing_ter}
	Let $\mathbf{Q} \triangleq \mathbf{A}^\top \mathbf{A} \in \mathbb{R}^{L \times L}$, then the optimization problem \eqref{eq:opt 1} can be equivalently solved by 
	\begin{equation} \label{eq:opt 1_final}
\begin{aligned}
	&\underset{\mathbf{Q},\mathbf{X}_\mathcal{S},\mathbf{Y}}{\text{minimize}} \hspace{0.25cm} \sum_\mathcal{S} \Tr \{\mathbf{X}_\mathcal{S}\}	& \\
	&\text{subject to} \hspace{0.25cm} \left[
\begin{array}{c c}
   \mathbf{R}^{-1} +\mathbf{E}_\mathcal{S}^\top \mathbf{H}^\top  (\frac{g^2}{\sigma_w^2}  \mathbf{Q} -  \mathbf{Y} ) \mathbf{HE}_\mathcal{S}   & \mathbf{I}_K \\ 
  \mathbf{I}_K  &   \mathbf{X}_\mathcal{S}   \\
\end{array}
\right] \succeq \mathbf{0} &\\
	&\hspace{1.7cm}  \left[
	\begin{array}{c c}
	   \mathbf{Y} & \frac{g}{\sigma_w}\mathbf{Q} \\ 
	  \frac{g}{\sigma_w}  \mathbf{Q}  &  \frac{\sigma_w^2}{g^2 \sigma_v^2} \mathbf{I}_L  + \mathbf{Q} \\
	\end{array}
	\right] \succeq \mathbf{0} , \; \mathcal{S} \subset  \Omega & \\
	 &\hspace{1.7cm} \Tr \{(\mathbf{HR}_x \mathbf{H}^\top + \sigma_v^2 \mathbf{I}_L) \mathbf{Q}\} \leq P & \\
	&\hspace{1.7cm} \mathbf{Q} \succeq \mathbf{0}  , \hspace{0.15cm} \mathrm{rank}(\mathbf{Q}) = M,&
\end{aligned}
\end{equation}
where the matrices  $\mathbf{Q}$, $\mathbf{X}_\mathcal{S} \in \mathbb{R}^{K \times K}$ and $\mathbf{Y} \in \mathbb{R}^{L \times L}$ are optimization variables. 
\end{theorem}

\begin{remark}
	The last two constraints in \eqref{eq:opt 1_final} appear due to the variable transformation $\mathbf{Q} = \mathbf{A}^\top \mathbf{A}$ which is a rank-$M$ positive semi-definite matrix. The difficulty of solving \eqref{eq:opt 1_final} is due to the rank constraint which makes the optimization problem non-convex in general. However, the constraint can be relaxed making the remaining problem  convex, a technique which is usually called semi-definite relaxation (SDR). Note also that the SDR can be only used to give a lower-bound on the optimal cost of the original objective function in \eqref{eq:opt 1_final}. 
\end{remark}

Next, we develop a \textit{two-stage} procedure in order to approximately solve for $\mathbf{A}$ in the non-convex optimization problem \eqref{eq:opt 1_final}. 

\textit{1) Semi-definite relaxation (SDR):} We first ignore the rank constraint in \eqref{eq:opt 1_final}, and solve the convex SDR problem for the matrix $\mathbf{Q}$. In some cases, closed-form solutions exist which we discuss later in this section. After finding the optimal $\mathbf{Q}^\star$, we take the eigen-value decomposition (EVD)
$\mathbf{Q}^\star = \mathbf{U}_q \mathbf{\Gamma}_q \mathbf{U}_q^\top$, where $\mathbf{U}_q \in \mathbb{R}^{L \times L}$ is a unitary matrix and $\mathbf{\Gamma}_q = \mathrm{diag} \left(\gamma_{q_1},  \ldots , \gamma_{q_L} \right) \in \mathbb{R}^{L \times L}$ such that $\gamma_{q_1} > \ldots > \gamma_{q_L}$. 

 \textit{2) Low-rank approximation:} We approximately reconstruct the rank--$M$ sensing matrix $\mathbf{A}$ by solving 
	\begin{equation} \label{eq:opt_rec_A_appx}
	\begin{aligned}
		&\underset{\mathbf{A}}{{\text{minimize}}} \hspace{0.25cm} \|   \mathbf{A}^\top \mathbf{A} - \mathbf{Q} \|_F^2.&
	\end{aligned}
	\end{equation}
It can be shown that the optimal sensing matrix $\mathbf{A}^\star$ (with respect to \eqref{eq:opt_rec_A_appx}) has the following structure \cite{11:Kokiopoulou}:

\begin{equation} \label{eq:recstr A}
	\mathbf{A}^\star = \mathbf{U}_a \left[\mathrm{diag}(\sqrt{\gamma_{q_1}} , \ldots , \sqrt{\gamma_{q_M}}) \;\; \mathbf{0}_{M \times (L-M)}\right] \mathbf{U}_q^\top,
\end{equation}
where $\mathbf{U}_a \! \in \! \mathbb{R}^{M�\! \times \! M}$ is an arbitrary unitary matrix. 

Since the eigenvalues $\gamma_{q_{M+1}},\! \ldots \!, \gamma_{q_N}$ are dropped in \eqref{eq:recstr A}, we finally scale the resulting $\mathbf{A}^\star$ by $\sqrt{P / \Tr\{(\mathbf{H} \mathbf{R}_x \mathbf{H}^\top \!+\! \sigma_v^2 \mathbf{I}_L)\mathbf{A}^{\star \top} \mathbf{A}^{\star}\}}$ in order to satisfy the power constraint by equality.

Next, we investigate the optimization problem \eqref{eq:opt 1_final} in several special cases. 

\subsubsection{Special Case I ($\mathbf{R} \! = \! \sigma_x^2 \mathbf{I}_K$, $\mathbf{H} \!=\! \mathbf{I}_N$)} \label{sec:special case1}

The motivation is to study a case where the non-zero components of the sparse source are uncorrelated, i.e., $\mathbf{R} \!=\! \sigma_x^2 \mathbf{I}_K$ and the source-to-sensor channel is only subject to additive noise, i.e., $\mathbf{H} \!=\! \mathbf{I}_N$. 
\begin{proposition} \label{prop:sepc_1}
	Let $\mathbf{R} = \sigma_x^2 \mathbf{I}_K$ and $\mathbf{H} = \mathbf{I}_N$, then the solution to the two-stage optimization procedure is given by
	\begin{equation} \label{eq:closed_spec1_4_proof}
		\mathbf{A}^\star = \sqrt{\frac{KP}{M(\sigma_x^2 + K \sigma_v^2)}} \; \mathbf{U}_a [\mathbf{I}_M \; \; \mathbf{0}_{M \times 		(N-M)}] \mathbf{V}_a^\top, 
	\end{equation}
	where $\mathbf{U}_a \in \mathbb{R}^{M \times M}$ and $\mathbf{V}_a \in \mathbb{R}^{N \times N}$ are arbitrary unitary matrices. 
\end{proposition}

\begin{remark}
	The structure of the sensing matrix in \eqref{eq:closed_spec1_4_proof} is normally referred to as \textit{tight frame}. Such structure is also optimal in certain cases, for example, the optimality of a tight frame-structured sensing matrix has been shown in \cite{12:Chen} with respect to minimizing LS-based oracle estimator. Another advantage of tight frames is that they can be efficiently constructed in finite number of arithmetic operations. 
\end{remark}

\subsubsection{Special Case II ($\mathbf{R} = \sigma_x^2 \mathbf{I}_K$, $\mathbf{v} = \mathbf{0}, \mathbf{H}:$ square full rank)} \label{sec:special case1}
Now, we discuss a case where the non-zero components of the sparse source are uncorrelated, i.e., $\mathbf{R} = \sigma_x^2 \mathbf{I}_K$, the observations before encoding are noiseless, i.e., $\mathbf{v=0}$ and $\mathbf{H}$ is a full-rank matrix. 
\begin{proposition} \label{prop:sepc_2}
	Let $\mathbf{R} = \sigma_x^2 \mathbf{I}_K$ and $\mathbf{v} = \mathbf{0}$, and consider that $\mathbf{H}$ is a square full-rank matrix such that its SVD can be written as $\mathbf{H} = \mathbf{U}_h \mathbf{\Gamma}_h  \mathbf{V}_h^\top$, where $ \mathbf{U}_h$ and  $\mathbf{V}_h$ are $N \times N$ unitary matrices and $\mathbf{\Gamma}_h = \mathrm{diag}(\gamma_{h_1},\gamma_{h_2}, \ldots, \gamma_{h_N} )$ is a diagonal matrix containing singular values $\gamma_{h_1}<\gamma_{h_2}< \ldots < \gamma_{h_N}$. Then, the solution to the two-stage optimization procedure is given by
	\begin{equation} \label{eq:closed_spec2_4}
		\mathbf{A}^\star = \sqrt{\frac{KP}{M\sigma_x^2 }} \; \mathbf{U}_a [\mathbf{\Gamma}_a \; \; \mathbf{0}_{M \times (N-M)}]  \mathbf{U}_h^\top, 
	\end{equation}
	where $\mathbf{U}_a \in \mathbb{R}^{M \times M}$ is an arbitrary unitary matrix, and $\mathbf{\Gamma}_a = \mathrm{diag}(\gamma_{h_1}^{-1},\ldots,\gamma_{h_M}^{-1})$. 
\end{proposition}

\begin{remark}
	According to \eqref{eq:closed_spec2_4} in \proref{prop:sepc_2}, the effective received measurement matrix at the decoder, i.e., $g \mathbf{AH}$, has a tight-frame structure. Interestingly, it can also be shown (see e.g. \cite{11:Duarte}) that the optimized sensing matrix derived in \eqref{eq:closed_spec2_4}, without the scaling factor, coincides with the optimal solution to the optimization problem
$\underset{\mathbf{A}}{{\text{minimize}}} \hspace{0.25cm} \| \mathbf{H}^	\top \mathbf{A}^\top \mathbf{A} \mathbf{H}  - \mathbf{I}_N \|_F ,$
which belongs to the second category of sensing matrix design problems introduced in \secref{sec:intro}.
	
\end{remark}

\subsubsection{Special Case III ($\mathbf{w= 0} $, $\mathbf{H} = \mathbf{I}_N$, $\mathbf{R} = \sigma_x^2 \mathbf{I}_K$)} \label{sec:special case3}
Here, we investigate a case where the additive channel noise in the system is negligible, i.e., $\mathbf{w=0}$,   the observations before encoding are only subject to additive noise, i.e., $\mathbf{H} = \mathbf{I}_N$, and the non-zero components of the sparse source vector are uncorrelated, i.e., $\mathbf{R} = \sigma_x^2 \mathbf{I}_K$. In this case, the optimal sensing matrix to the original problem \eqref{eq:opt 1} is given by the following proposition. 
\begin{proposition} \label{prop:sepc_3}
	Let $\mathbf{w= 0} $, $\mathbf{H} = \mathbf{I}_N$, $\mathbf{R} = \sigma_x^2 \mathbf{I}_K$, then, the solution to the optimization problem \eqref{eq:opt 1} is given by
	\begin{equation} \label{eq:spec_case_3}
	\mathbf{A}^\star = \sqrt{\frac{KP}{M \sigma_x^2}} \mathbf{U}_a [\mathbf{I}_{M} \: \: \mathbf{0}_{M \times (N-M)}],
\end{equation}
where $\mathbf{U}_a \in \mathbb{R}^{M \times M}$ is an arbitrary unitary matrix.  
\end{proposition}

\subsubsection{Special Case IV ($\mathbf{v= 0} $, $\frac{g^2}{\sigma_w^2} \rightarrow 0$)} \label{sec:special case4}
Now, we consider an asymptotic case, where the communication channel is in a noisy regime such that the ratio between the power of channel gain over the power of additive channel noise tends to zero, i.e., $g^2/\sigma_w^2 \rightarrow 0$. 
\begin{proposition} \label{prop:sepc_4}
	Let $\mathbf{v= 0} $ and $\frac{g^2}{\sigma_w^2} \rightarrow 0$, and define $\mathbf{T} \triangleq \sum_\mathcal{S}  \mathbf{D}_\mathcal{S} \mathbf{R}^2 \mathbf{D}_\mathcal{S}^\top$ and $\mathbf{Z} \triangleq \mathbf{T}^{-1/2} \mathbf{HR}_x \mathbf{H}^\top \mathbf{T}^{-1/2}$ which has the EVD $\mathbf{Z} = \mathbf{U}_z \mathbf{\Gamma}_z \mathbf{U}_z^\top$. Then, the approximate solution to the two-stage optimization procedure is asymptotically given by
	\begin{equation} \label{eq:sol asymp2}
		\mathbf{A}^\star =  \mathbf{U}_a \left[\mathrm{diag}  \left(\sqrt{\gamma_q},0,\ldots,0 \right)\; \; \mathbf{0}_{M \times (L-M)} 		\right] \mathbf{U}_q^\top,
	\end{equation}
	where $\mathbf{U}_a \in \mathbb{R}^{M \times M}$ is an arbitrary unitary matrix,  and $\gamma_q$ is the only non-zero eignevlaue of 
	\begin{equation} \label{eq:sol asymp}
		\mathbf{Q}^\star =  \mathbf{T}^{-1/2}\mathbf{U}_z \mathrm{diag} \left(\frac{P}{\gamma_{z_1}},0,\ldots,0 \right)\mathbf{U}_z^\top \mathbf{T}^{-1/2}.
	\end{equation}
	Further, $\mathbf{U}_q$ is the eigenvector associated with the EVD of $\mathbf{Q}^\star$, and $\gamma_{z_1}$ is the smallest eigenvalue of $\mathbf{Z}$.
	\end{proposition}
\begin{remark}
From \eqref{eq:sol asymp2}, it can be observed if channel condition degrades, as $g^2 / \sigma_w^2 \rightarrow 0$, the approximate sensing matrix has only one active singular-value.
\end{remark}

\vspace{-0.25cm}
\section{Complexity Considerations} \label{sec:complexity}
Here, we discuss the computational complexity of solving the proposed optimization scheme for sensing matrix design, and offer an approach in order to solve the optimization problem with significantly less computational complexity.

The high computational complexity in the two-stage optimization procedure proposed in \secref{sec:design} arises from the first step, i.e., solving the SDR problem (\eqref{eq:opt 1_final}  without the rank constraint). More precisely, the optimization problem consists of one matrix variable $\mathbf{Q}$ of size $L \times L$,  ${N \choose K}$  matrix variables $\mathbf{X}_\mathcal{S}$ of size $K \times K$, and one matrix variable $\mathbf{Y}$ of size $L \times L$. Hence, it can be iteratively solved using interior point methods with computational complexity growing at most like $\mathcal{O}(2 L^6 + {N \choose K}^3 K^6)$ arithmetic operations in each iteration \cite{06:Zhi-Quan}. Therefore, as $N$ increases, the computational complexity grows exponentially due to the term ${N \choose K}$. 

The computational complexity of solving the SDR problem can be significantly reduced under certain assumptions (see, e.g., the special cases I-IV), in which closed-form solutions can be derived. Here, we offer an alternative in order to solve the SDR problem of \eqref{eq:opt 1_final} with a reduced  computational complexity. Note that the objective function $\mathrm{MSE}^{(lb)}$ in \eqref{eq:oracle_MSE_2} can be rewritten as 
\begin{equation} \label{eq:less_complex}
	\mathrm{MSE}^{(lb)} = \mathbb{E}_{\boldsymbol{\mathcal{S}}} \left[  \Tr \left\{ \left(\mathbf{R}^{-1} \!+\! g^2 \mathbf{E}_{\boldsymbol{\mathcal{S}}}^\top \mathbf{H}^\top \mathbf{A}^\top \mathbf{R}_n^{\! -1} \mathbf{A} \mathbf{H} \mathbf{E}_{\boldsymbol{\mathcal{S}}} \right)^{\!-1} \right\} \right],
\end{equation}
where $\boldsymbol{\mathcal{S}}$ is a random variable which picks a support set $\mathcal{S}$ uniformly at random from the set of all possibilities $\Omega$, and $\mathbb{E}_{\boldsymbol{\mathcal{S}}}$ means that the expectation is taken only over the randomness of $\boldsymbol{\mathcal{S}}$. Notice that the expectation in \eqref{eq:less_complex} can be (approximately) calculated using sample mean as
\begin{equation} \label{eq:less_complex_2}
	\mathrm{MSE}^{(lb)} \! \approx \! \!\frac{1}{|\Omega'|} \! \sum_{\mathcal{S}' \in \Omega'} \!   \Tr \! \left\{ \! \left(\mathbf{R}^{-1} \! \!+\! g^2 \mathbf{E}_{\mathcal{S}'}^\top \mathbf{H}^\top \mathbf{A}^\top \mathbf{R}_n^{\! -1} \mathbf{A} \mathbf{H} \mathbf{E}_{\mathcal{S}'} \! \right)^{\!-1} \! \right\} 
\end{equation}
where $\mathcal{S}'$ is uniformly chosen from a set $\Omega'$, which is a subset of $\Omega$. Note that the cardinality $|\Omega'|$ can be chosen to be far less than ${N \choose K}$ and still obtain a good approximation of 
(\ref{eq:less_complex}).
As a result, the computational complexity of solving the resulting SDR problem reduces to $\mathcal{O}(2 L^6 + |\Omega'|^3 K^6)$ arithmetic operations, where $ |\Omega'| \ll {N \choose K}$. 


\vspace{-0.25cm}
\section{Numerical Experiments} \label{sec:sim}
Now, we provide numerical experiments for evaluating the performance of the proposed sensing matrix design scheme in \secref{sec:design}, referred to as \textit{lower-bound minimizing sensing matrix}. We compare it with the following design methods:
\begin{itemize}[leftmargin=*]
	 \item \textit{Upper-bound minimizing sensing matrix:} Using this method, which has been studied in \cite{08:Jin,07:Schizas} in non-CS framework, the MSE of the MMSE estimator of the sparse source vector is upper-bounded by that of the linear MMSE (LMMSE) estimator. The MSE incurred by using the LMMSE estimator can be written as \cite{93:Kay}
\begin{equation} \label{eq:ub_MSE}
\begin{aligned}
	\mathrm{MSE}^{(ub)} &\triangleq  \Tr \left\{ \left(\mathbf{R}_x^{-1} + g^2 \mathbf{H}^\top \mathbf{A}^\top \mathbf{R}_n^{-1} \mathbf{A} \mathbf{H} \right)^{-1} \right\},&
\end{aligned}
\end{equation}	
which is minimized subject to a power constraint. 
	\item \textit{Gaussian sensing matrix:} Using this standard method, each element of the Gaussian sensing matrix is generated randomly according to the standard Gaussian distribution. 
	\item \textit{Tight frame}: Using this method, the sensing matrix is chosen as $\mathbf{A} =  \mathbf{U}_a \left[\mathbf{I}_M \; \; \mathbf{0}_{M \times (L-M)}\right]  \mathbf{V}_a^\top$, where $\mathbf{U}_a \in \mathbb{R}^{M \times M}$ and $\mathbf{V}_a \in \mathbb{R}^{L \times L}$ are arbitrary unitary matrices. 
\end{itemize}

Note that we scale the resulting sensing matrix, described above, by $\sqrt{P / \Tr\{(\mathbf{H} \mathbf{R}_x \mathbf{H}^\top+ \sigma_v^2 \mathbf{I}_L)\mathbf{A}^{\top} \mathbf{A}\}}$ in order to satisfy the power constraint. We also compare the actual MSE, incurred by using the above methods, with the value of the lower-bound \eqref{eq:oracle_MSE_2} when the lower-bound sensing matrix is applied. This will be referred to as \textit{lower-bound} in our experiments. It should be also mentioned that for solving the convex SDR problems, we use the \texttt{CVX} solver \cite{cvx} . 


We evaluate the performance using the normalized MSE (NMSE) criterion, defined as 
\begin{equation*}
	\mathrm{NMSE} \triangleq \frac{\mathbb{E}[\|\mathbf{x} - \widehat{\mathbf{x}}\|_2^2]}{K}.
\end{equation*}

We randomly generate a set of exactly $K$-sparse vector $\mathbf{x}$, where the support set $\mathcal{S}$ with $|\mathcal{S}| = K$ is chosen uniformly at random over the set $\{1,2,\ldots,N\}$. The non-zero components of $\mathbf{x}$ are drawn from Gaussian distribution $\mathcal{N}(\mathbf{0}_K,\mathbf{R})$. The covariance matrix $\mathbf{R}$ is generated according to the exponential model \cite{01:Loyka}, where each entry at row $i$ and column $j$ is chosen as $\rho^{| i - j |}$ in which $0 \leq \rho < 1$ is known as the correlation coefficient. We Compute sample covariance matrix for the sparse source vector, i.e., $\mathbf{R}_x = \mathbb{E}[\mathbf{x} \mathbf{x}^\top]$ using $10^5$ random generation of the source vector $\mathbf{x}$. We also let $\mathbf{v= 0}$ and $\mathbf{H} = \mathbf{I}_N$. In order to implement the decoder, we use two different sparse reconstruction methods: the greedy orthogonal matching pursuit (OMP) algorithm \cite{07:Tropp}, and the Bayesian-based random--OMP algorithm \cite{09:Elad}, which is a low-complexity approximation of the exact (exhaustive) MMSE estimator. The actual performance of the proposed design methods is assessed using Monte-Carlo simulations by generating $5000$ realizations of the input sparse vector $\mathbf{x}$. In our first two experiments, we use, at the decoder, random-OMP algorithm for reconstruction of sparse source vector. 

In our first experiment, we use the simulation parameters $N = 36, K= 3, P=10 \text{ dB}, g = 0.5, \sigma_w = 0.1, \rho = 0.25$. We plot the NMSE of the design methods as a function of $M$ in \figref{fig:MSE_M_randOMP_new}, and observe that at all measurement regions, the proposed lower-bound minimizing sensing matrix outperforms the other competing methods by taking into account the sparsity pattern of the  source. As expected, as the number of measurements increases, the performance of the methods improves, however, it finally saturates and increasing $M$ further does not help to improve NMSE. This is because at higher number of measurements, the NMSE is influenced more by the additive noise in the system which is fixed. As $M$ increases, the performance of the tight frame approaches that of the lower-bound  minimizing sensing matrix, which shows that the latter behaves like an orthogonal transform. 

\begin{figure}
  \begin{center}
  \includegraphics[width=0.75\columnwidth,height=5.2cm]{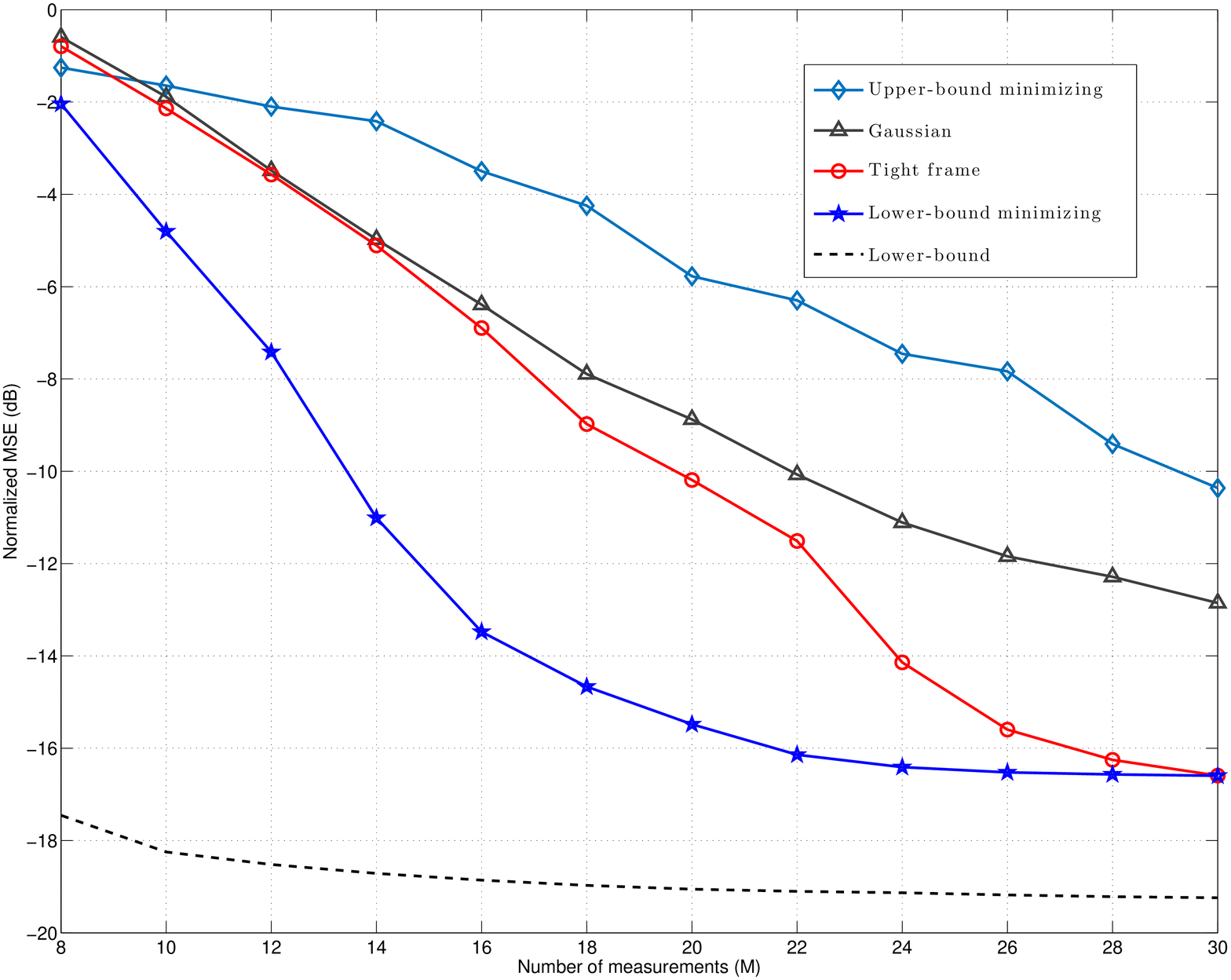}\\
  \caption{NMSE (in dB) as a function of number of measurements $M$. }
  \label{fig:MSE_M_randOMP_new}
  \end{center}
   \vspace{-0.35cm}
\end{figure}

Using the same simulation parameters, by fixing $M = 18$, we vary the transmission power $P$ (in dB), and evaluate the performance of the methods in terms of NMSE. The results are reported in \figref{fig:MSE_P_randOMP_new}. At the low power regime, the performance of the competing methods are almost the same, however, as $P$ increases, the proposed lower-bound minimizing sensing matrix outperforms the other schemes. For example, at $P = 10$ dB, the proposed scheme gives a better performance by more than 6 dB as compared to the other methods.  

\begin{figure}
  \begin{center}
 \includegraphics[width=0.75\columnwidth,height=5.2cm]{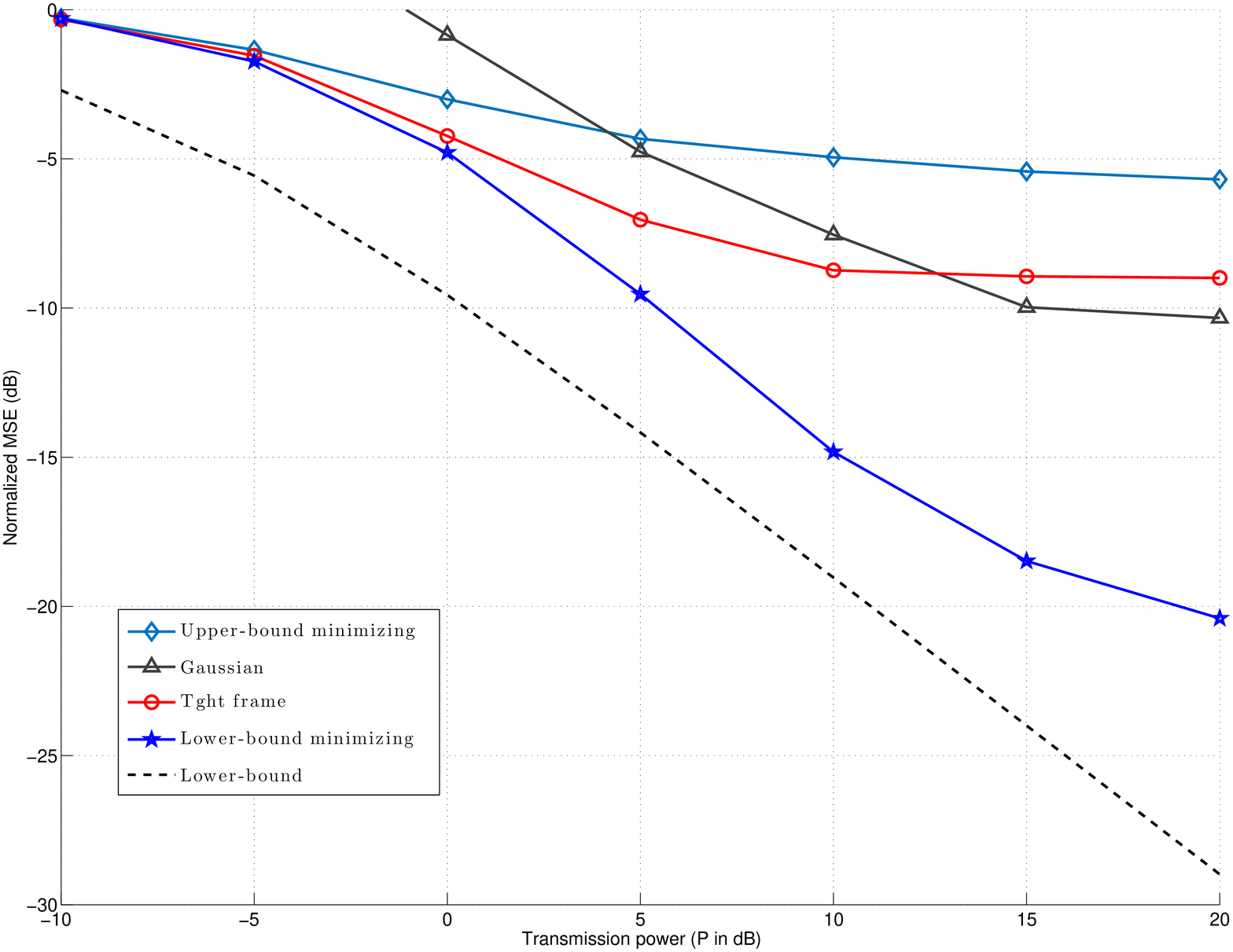}\\
  \caption{NMSE (in dB) as a function of transmission power $P$ (in dB).}
  \label{fig:MSE_P_randOMP_new}
  \end{center}
     \vspace{-0.4cm}
\end{figure}

In the previous experiments, we have used the random-OMP algorithm (as the approximation of the exact MMSE estimator) for reconstructing the sparse source. While this algorithm is nearly optimal (in MSE sense), the reconstructed vector might not be necessarily a sparse vector \cite{09:Elad}. In some applications, together with reconstruction accuracy, one might desire a sparse representation at the receiving-end. Therefore, we use the greedy OMP algorithm \cite{07:Tropp} which preserves the sparse structure through reconstructing the source.

We compare the performance of the methods (in terms of NMSE) as a function of channel signal to noise ratio (CSNR), defined as $\mathrm{CSNR} \triangleq g^2 / \sigma_w^2$, in logarithmic scale. The results are reported in \figref{fig:MSE_G_OMP_new}. Simulation parameters are chosen as $N = 36, K = 3, P = 10 \text{ dB}, M = 18, \rho = 0.5$. We fix $\sigma_w = 0.1$, and vary the CSNR from $1$ to $10^3$ where the channel gains $g$ are chosen accordingly. It is observed that at $\mathrm{CSNR} = 10^2$, the lower-bound minimizing sensing matrix outperforms the Gaussian sensing matrix by more than 8 dB, and the upper-bound minimizing sensing matrix by more than 10 dB. 

\begin{figure}
  \begin{center}
 \includegraphics[width=0.75\columnwidth,height=5.2cm]{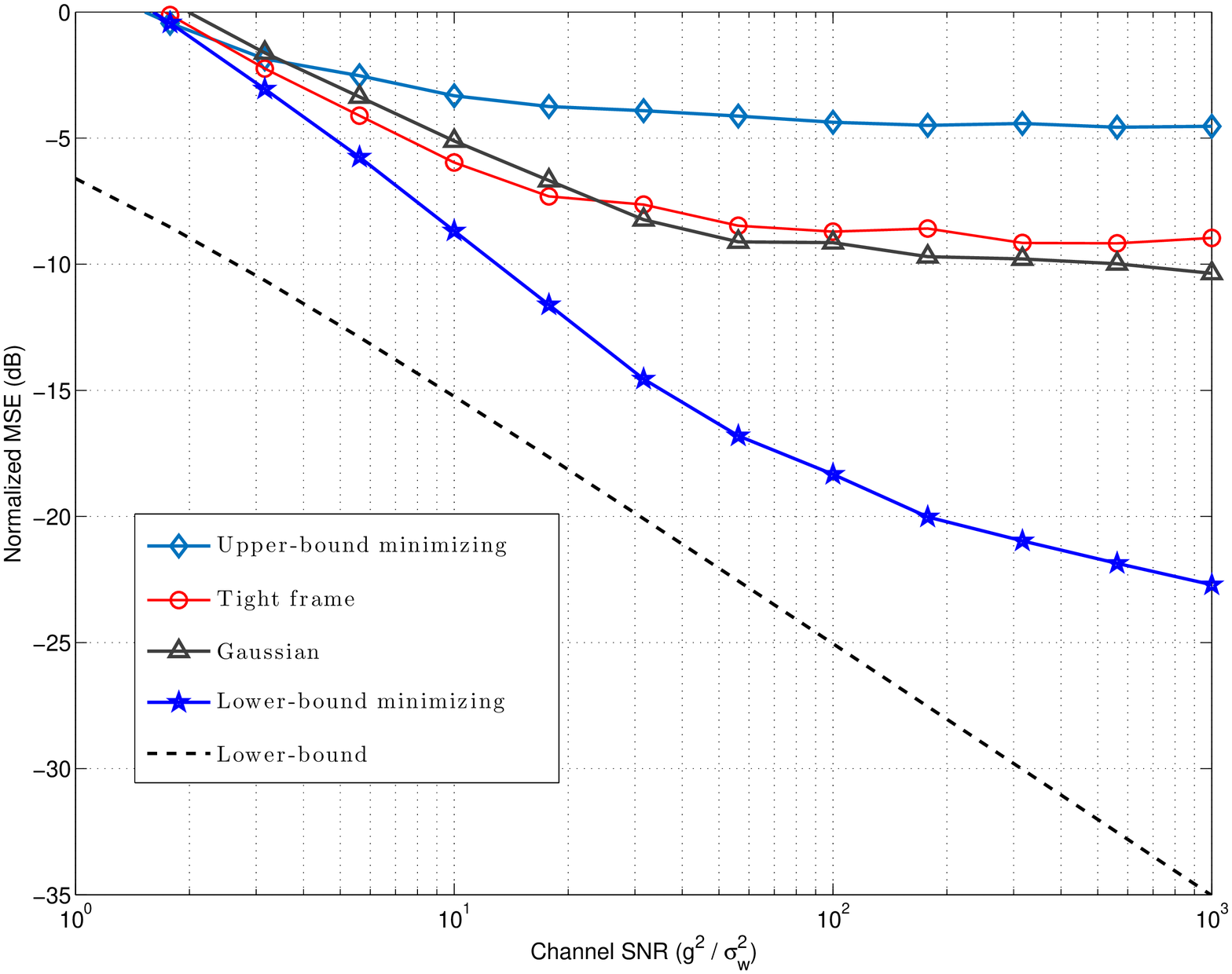}\\
  \caption{NMSE (in dB) as a function of CSNR (in dB).}
  \label{fig:MSE_G_OMP_new}
  \end{center}
     \vspace{-0.4cm}
\end{figure}

In our final experiment, we implement a higher-dimensional system, and apply the proposed low-complexity approach introduced in \secref{sec:complexity}. For this purpose, we choose the following simulation parameters: $N = 100, K = 5, \sigma_w = 0.1, g = 0.5, P =10 \text{ dB}, \rho = 0.75$, and plot the NMSE by varying $M$ in \figref{fig:MSE_M_100_new}. Further, the cardinality of the set $\Omega'$ in \eqref{eq:less_complex_2} is set to 2500, while the cardinality of the set of all sparsity patterns is $| \Omega |  = {N \choose K} \approx 7.5 \times 10^7$. It can be observed while the computational complexity of the lower-bound minimizing scheme has been considerably reduced, it still outperforms the other methods.

To observe the efficiency of the low-rank approximation in the second stage of our proposed method, we also show the performance of another design method, in the same figure, labelled by `randomization', where we use the randomization technique from \cite{10:Zhi} instead of the second stage in our method, given by \eqref{eq:recstr A}. More precisely, using this method, we assume that the resulting sensing matrix is given by $\mathbf{A} = \mathbf{V} \mathbf{\Gamma}^{1/2} \mathbf{U}_q^\top$, where $\mathbf{V} \in \mathbb{R}^{M \times L}$ is a random matrix whose elements $[\mathbf{V}]_{ij}$ are drawn from $\mathcal{N}(0, 1/ \sqrt{M})$ such that $\mathbb{E}[\mathbf{A}^\top \mathbf{A}] = \mathbf{Q}$. Note that we rescale each realization of $\mathbf{A}$ to meet the power constraint, and choose the one (out of 1000 realizations) which gives the lowest $\mathrm{MSE}^{(lb)}$. 

\begin{figure}
  \begin{center}
 \includegraphics[width=0.75\columnwidth,height=5.2cm]{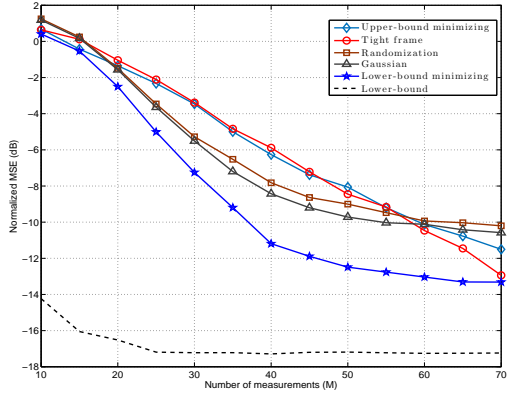}\\
  \caption{NMSE (in dB) as a function of number of measurements $M$.}
  \label{fig:MSE_M_100_new}
  \end{center}
     \vspace{-0.4cm}
\end{figure}

\vspace{-0.35cm}
\section{Conclusions} \label{sec:conclusions}
\vspace{-0.03cm}
We have proposed an optimization procedure for designing the sensing matrix, under a power constraint, in the CS framework. The design aims to minimize a lower-bound on MSE of sparse source reconstruction. Under certain conditions, we have been able to address the optimization procedure by deriving closed-form expressions for the sensing matrix. Numerical results show the advantage of our proposed design compared to other relevant schemes. This advantage has been achieved at the price of higher computational complexity. Therefore, we proposed an alternative approximation to the MSE lower bound objective function that results in significant complexity reduction.
\vspace{-0.3cm}

\begin{appendices}
\section{Some Useful Lemmas} \label{sec:lemmas}

The first lemma is straightforward to prove.
\begin{lemma} \label{lem:prop}
Given the matrix $\mathbf{E}_\mathcal{S} \in \mathbb{R}^{N \times K}$, which is constructed as an identity matrix whose columns excluding those indexed by the support set $\mathcal{S}$ are deleted, it has the following properties:
\begin{itemize}
	\item $\mathbf{E}_\mathcal{S}^\top \mathbf{E}_\mathcal{S} = \mathbf{I}_K$,
	\item $\sum_{\mathcal{S}} \mathbf{E}_\mathcal{S} \mathbf{E}_\mathcal{S}^\top = \frac{{N \choose K}}{K} \mathbf{I}_N.$
\end{itemize}
\end{lemma}

\begin{lemma} \label{lem:deriv_cov}
The covariance matrix of the sparse source, i.e., $\mathbf{R}_{x}$, can be parametrized as
\begin{equation} \label{deriv_cov}
\begin{aligned}
	\mathbf{R}_x = \frac{1}{{N \choose K}}\sum_\mathcal{S} \mathbf{E}_\mathcal{S} \mathbf{R} \mathbf{E}_\mathcal{S}^\top.
\end{aligned}
\end{equation}
\end{lemma}

\begin{proof}
We have
\begin{equation*}
\begin{aligned}
\mathbf{R}_x &= \mathbb{E}[\mathbf{xx}^\top] \overset{(a)}{=} \sum_\mathcal{S} p(\mathcal{S}) \mathbb{E}[\mathbf{xx}^\top \big | \mathcal{S}] & \\
	&\overset{(b)}{=} \frac{1}{{N \choose K}}\sum_\mathcal{S} \mathbb{E}[\mathbf{E}_\mathcal{S} \mathbf{x}_\mathcal{S} \mathbf{x}_\mathcal{S}^\top \mathbf{E}_\mathcal{S}^\top] = \frac{1}{{N \choose K}}\sum_\mathcal{S} \mathbf{E}_\mathcal{S} \mathbf{R} \mathbf{E}_\mathcal{S}^\top& 
\end{aligned}
\end{equation*}
where $(a)$ follows by marginalization over all support sets. Also, $(b)$ holds due to the fact that $p(\mathcal{S}) = 1/ {N \choose K}$; further, given a support set $\mathcal{S}$, $\mathbf{x} = \mathbf{E}_\mathcal{S} \mathbf{x}_\mathcal{S}$.
\end{proof}

\begin{lemma}\cite[page 249]{11:Marshall} \label{lem:constr bound}
	Let $\mathbf{A}$ and $\mathbf{B}$ are two $N \times N$ symmetric matrices, whose eigen-values $\alpha_1,\ldots, \alpha_N$ and $\beta_1,\ldots,\beta_N$ are ordered increasingly and decreasingly, respectively. Then $\Tr\{{\mathbf{AB}}\} \geq \sum_{i=1}^N \alpha_i \beta_i$.
\end{lemma}


\section{Proof of \theoref{theo:sing_ter}} \label{sec:app A}

To be able to solve the optimization problem in \eqref{eq:opt 1}, let us first define
\begin{equation} \label{eq:MSE_S}
	\mathrm{MSE}^{(lb)}_\mathcal{S} \triangleq \Tr \left\{\left(\mathbf{R}^{-1} + g^2 \mathbf{E}_\mathcal{S}^\top \mathbf{H}^\top \mathbf{A}^\top \mathbf{R}_n^{-1} \mathbf{A} \mathbf{H} \mathbf{E}_\mathcal{S} \right)^{-1} \right\}.
\end{equation}
Using matrix inversion lemma for $\mathbf{R}_n^{-1}$, we obtain
\begin{equation} \label{eq:mtx inv}
\begin{aligned}
	\mathbf{R}_n^{-1}  
	&= \sigma_w^{-2} \mathbf{I}_M - \sigma_w^{-2} \mathbf{A} \left(\frac{\sigma_w^{2}}{g^2 \sigma_v^{2}} \mathbf{I}_N + \mathbf{A}^\top \mathbf{A} \right)^{-1} \mathbf{A}^\top.& 
\end{aligned}
\end{equation}
Plugging \eqref{eq:mtx inv} back into \eqref{eq:MSE_S}, it follows that
\begin{equation} \label{eq:opt 1_reform1}
\begin{aligned}
	&\mathrm{MSE}^{(lb)}_\mathcal{S} = \Tr \left\{\bigg(\mathbf{R}^{-1} + \frac{g^2}{\sigma_w^2} \mathbf{E}_\mathcal{S}^\top \mathbf{H}^\top \mathbf{A}^\top \mathbf{A} \mathbf{H} \mathbf{E}_\mathcal{S} \right.& \\
	&  \left. - \frac{g^2}{\sigma_w^2} \mathbf{E}_\mathcal{S}^\top \mathbf{H}^\top \mathbf{A}^\top  \mathbf{A}\left(\frac{\sigma_w^2}{g^2 \sigma_v^2} \mathbf{I}_N + \mathbf{A}^\top \mathbf{A}\right)^{-1} \mathbf{A}^\top \mathbf{AHE}_\mathcal{S} \bigg)^{-1} \right\}.& 
\end{aligned}
\end{equation}

Next, defining $\mathbf{Q} \triangleq \mathbf{A}^\top \mathbf{A}$ and $\mathbf{D}_\mathcal{S} \triangleq \mathbf{HE}_\mathcal{S}$, the original optimization problem in \eqref{eq:opt 1} for finding optimized sensing matrix $\mathbf{A}$ can be equivalently translated into\footnote[1]{Note that since $p(\mathcal{S}) = 1/{N \choose K}$, it can be ignored in formulating the resulting optimization problems.}

\begin{equation} \label{eq:opt 1_reform2}
\begin{aligned}
	&\underset{\mathbf{Q}}{\text{minimize}} \hspace{0.25cm} \sum_\mathcal{S} \Tr \left\{\bigg(\mathbf{R}^{-1} + \frac{g^2}{\sigma_w^2} \mathbf{D}_\mathcal{S}^\top \mathbf{Q} \mathbf{D}_\mathcal{S}  \right.& \\
	 &\left. \hspace{2cm}- \frac{g^2}{\sigma_w^2} \mathbf{D}_\mathcal{S}^\top  \mathbf{Q} \left(\frac{\sigma_w^2}{g^2 \sigma_v^2} \mathbf{I}_N + \mathbf{Q} \right)^{-1} \mathbf{Q} \mathbf{D}_\mathcal{S} \bigg)^{-1} \right\}& \\
	& \text{subject to} \hspace{0.25cm} \Tr \{(\mathbf{HR}_x \mathbf{H}^\top + \sigma_v^2 \mathbf{I}_N) \mathbf{Q}\} \leq P & \\
	&\hspace{1.7cm} \mathbf{Q} \succeq \mathbf{0} , \hspace{0.15cm} \mathrm{rank}(\mathbf{Q}) = M,&
\end{aligned}
\end{equation}
where the rank constraint appears due to the fact that the sensing matrix has less rows than columns. Introducing the semidefinite slack variable matrix $\mathbf{X}_\mathcal{S} \in \mathbb{R}^{K \times K}$, we can alternatively solve
\begin{equation} \label{eq:opt 1_reform3}
\begin{aligned}
	&\underset{\mathbf{Q},\mathbf{X}_\mathcal{S}}{\text{minimize}} \hspace{0.25cm} \sum_\mathcal{S} \Tr \{\mathbf{X}_\mathcal{S}\}	& \\
	&\text{subject to} \hspace{0.25cm} \bigg(\mathbf{R}^{-1} + \frac{g^2}{\sigma_w^2} \mathbf{D}_\mathcal{S}^\top \mathbf{Q} \mathbf{D}_\mathcal{S}  & \\
	 &\hspace{1.7cm} \!-\! \frac{g^2}{\sigma_w^2} \mathbf{D}_\mathcal{S}^\top  \mathbf{Q}\big(\frac{\sigma_w^2}{g^2 \sigma_v^2} \mathbf{I}_N\! + \!\mathbf{Q}\big)^{-\! 1} \mathbf{Q} \mathbf{D}_\mathcal{S} \bigg)^{- \!1} \! \! \! \! \! \preceq  \! \mathbf{X}_\mathcal{S}, \mathcal{S} \! \subset  \! \Omega& \\
	 &\hspace{1.7cm} \Tr \{ (\mathbf{HR}_x \mathbf{H}^\top + \sigma_v^2 \mathbf{I}_N) \mathbf{Q}\} \leq P & \\
	&\hspace{1.7cm} \mathbf{Q} \succeq \mathbf{0} , \hspace{0.15cm} \mathrm{rank}(\mathbf{Q}) = M.&
\end{aligned}
\end{equation}

Next, by applying the Schur's complement \cite{04:Boyd_book}, the first constraint in \eqref{eq:opt 1_reform3} can be rewritten as the positive semi-definite constraint $\mathbf{C} \succeq \mathbf{0}$, where $\mathbf{C} \triangleq $
\begin{equation} \label{eq:opt 1_reform4}
\begin{aligned}
&\left[\! \! \!
\begin{array}{c c}
   \mathbf{R}^{\!-1} \!+ \!\frac{g^2}{\sigma_w^2} \mathbf{D}_\mathcal{S}^\top \mathbf{Q} \mathbf{D}_\mathcal{S}  \!-\! \frac{g^2}{\sigma_w^2} \mathbf{D}_\mathcal{S}^\top  \mathbf{Q}(\frac{\sigma_w^2}{g^2 \sigma_v^2} \mathbf{I}_N \!+\! \mathbf{Q})^{-1} \mathbf{Q} \mathbf{D}_\mathcal{S} & \mathbf{I}_K \\ 
  \mathbf{I}_K  &   \mathbf{X}_\mathcal{S}   \\
\end{array} \! \! \!
\right] & 
\end{aligned}
\end{equation}

Introducing another slack semidefinite variable matrix $\mathbf{Y} \in \mathbb{R}^{N \times N}$, such that $\mathbf{Y} \succeq \frac{g^2}{\sigma_w^2}  \mathbf{Q}(\frac{\sigma_w^2}{g^2 \sigma_v^2} \mathbf{I}_N \!+\! \mathbf{Q})^{-1} \mathbf{Q}$, and using the Schur's complement for the resulting matrix inequality, we can further decompose the constraint in \eqref{eq:opt 1_reform4} into two linear matrix inequalities as follows

\begin{equation} \label{eq:opt 1_reform5}
\left[
\begin{array}{c c}
   \mathbf{R}^{-1} + \frac{g^2}{\sigma_w^2} \mathbf{D}_\mathcal{S}^\top  \mathbf{Q} \mathbf{D}_\mathcal{S}  - \mathbf{D}_\mathcal{S}^\top\mathbf{Y} \mathbf{D}_\mathcal{S} & \mathbf{I}_K \\ 
  \mathbf{I}_K  &   \mathbf{X}_\mathcal{S}   \\
\end{array}
\right] \succeq \mathbf{0},
\end{equation}

\begin{equation} \label{eq:opt 1_reform6}
\left[
\begin{array}{c c}
   \mathbf{Y} & \frac{g}{\sigma_w}\mathbf{Q} \\ 
  \frac{g}{\sigma_w} \mathbf{Q}  &  \frac{\sigma_w^2}{g^2 \sigma_v^2} \mathbf{I}_N  + \mathbf{Q} \\
\end{array}
\right] \succeq \mathbf{0}.
\end{equation}

Thus, using the linear matrix inequality constraints \eqref{eq:opt 1_reform5} and \eqref{eq:opt 1_reform6} in \eqref{eq:opt 1_reform3}, we can solve the optimization problem which is expressed by \eqref{eq:opt 1_final}.

Note that the optimal $\mathbf{Q}^\star$ is a rank--$M$ matrix, and has to have $M$ non-zero eigen-values, otherwise the optimal value of the objective function cannot be obtained. Now, let the EVD of $\mathbf{Q}^\star$ be $\mathbf{U}_q \mathrm{diag}(\mathbf{\Gamma}_q, \mathbf{0}_{(L-M)\times (L-M)}) \mathbf{U}_q^\top$, where $\mathbf{\Gamma}_q$ is a diagonal matrix containing the non-zero eigen-values. Then, noting $\mathbf{Q}^\star = \mathbf{A}^{\star \top} \mathbf{A}^\star$, the optimal $\mathbf{A}^\star$ becomes $\mathbf{A}^\star = \mathbf{U}_a [\mathbf{\Gamma}_q^{1/2} \; \; \mathbf{0}_{M \times (L-M)}] \mathbf{U}_q^\top$. Hence, the optimal sensing matrix $\mathbf{A}^\star$ minimizes \eqref{eq:opt 1}.

\section{Proof of \proref{prop:sepc_1}}
Using the notation $\mathbf{Q} = \mathbf{A}^\top \mathbf{A}$, we rewrite  \eqref{eq:opt 1_reform1} as
\begin{equation} \label{eq:closed_spec1_1}
\begin{aligned}
	&\mathrm{MSE}^{(lb)}_\mathcal{S} = \Tr \left\{\bigg(\frac{1}{\sigma_x^2} \mathbf{I}_K + \frac{g^2}{\sigma_w^2} \mathbf{E}_\mathcal{S}^\top \mathbf{Q} \mathbf{E}_\mathcal{S}  \right.& \\
	& \hspace{1.4cm} - \left. \frac{g^2}{\sigma_w^2} \mathbf{E}_\mathcal{S}^\top \mathbf{Q} \left(\frac{\sigma_w^2}{g^2 \sigma_v^2} \mathbf{I}_N + \mathbf{Q}\right)^{-1} \mathbf{Q} \mathbf{E}_\mathcal{S} \bigg)^{-1} \right\}.& 
\end{aligned}
\end{equation}

Applying \lemref{lem:deriv_cov}, the power constraint can be written as 
\begin{equation} \label{eq:power_cons_spec1}
\begin{aligned}
	&\Tr \left\{ \left(\frac{1}{{N \choose K}} \sum_\mathcal{S} \mathbf{E}_\mathcal{S} \mathbf{E}_\mathcal{S}^\top  + \sigma_v^2 \mathbf{I}_N \right) \mathbf{Q} \right\} & \\
	&\overset{(a)}{=}   \Tr \left\{ \left(\frac{\sigma_x^2}{K}   +  \sigma_v^2 \right) \mathbf{Q} \right\} \leq P,&
\end{aligned}
\end{equation}
where $(a)$ follows from \lemref{lem:prop}.

Considering the objective function $\sum_\mathcal{S} \mathrm{MSE}^{(lb)}_\mathcal{S}$, it can be lower-bounded as
\begin{equation} \label{eq:closed_spec1_2}
\begin{aligned}
	&\sum_\mathcal{S} \mathrm{MSE}^{(lb)}_\mathcal{S} \geq \sum_\mathcal{S} K^2 \big /  \Tr \left\{\bigg(\frac{1}{\sigma_x^2} \mathbf{I}_K + \frac{g^2}{\sigma_w^2} \mathbf{E}_\mathcal{S}^\top \mathbf{Q} \mathbf{E}_\mathcal{S}  \right.& \\
	& \hspace{1.4cm} - \left. \frac{g^2}{\sigma_w^2} \mathbf{E}_\mathcal{S}^\top \mathbf{Q} \left(\frac{\sigma_w^2}{g^2 \sigma_v^2} \mathbf{I}_N + \mathbf{Q}\right)^{-1} \mathbf{Q} \mathbf{E}_\mathcal{S} \bigg)^{-1} \right\},&
\end{aligned}
\end{equation}
where we used the inequality $\Tr\{\mathbf{B}^{-1}\} \geq \frac{K^2}{\Tr\{\mathbf{B}\}}$ for a positive definite matrix $\mathbf{B}$ of dimension $K \times K$ \cite[Lemma 2]{03:Shengli}, in which the equality is satisfied when $\mathbf{B}$ becomes a scaled identity matrix. Hence, the objective function in the left hand side of \eqref{eq:closed_spec1_2} reaches its minimum when $\mathbf{Q} = \alpha \mathbf{I}_N$ (for some $\alpha > 0$) since $\mathbf{E}_\mathcal{S}^\top  \mathbf{E}_\mathcal{S} = \mathbf{I}_K$ (cf. \lemref{lem:prop}), and the matrix inside the argument of the trace becomes an scaled identity matrix. Note that this choice of $\mathbf{Q}$ does not affect the power constraint. Further, the coefficient $\alpha$ is derived such that the constraint  \eqref{eq:power_cons_spec1} is satisfied with equality that yields $\alpha = \frac{KP}{N(\sigma_x^2 + K \sigma_v^2)}$. Therefore, assuming $\mathbf{R} = \sigma_x^2 \mathbf{I}_K$ and  $\mathbf{H} = \mathbf{I}_N$, the solution to the SDR problem is 
\begin{equation} \label{eq:closed_spec1_3}
	\mathbf{Q}^\star = \frac{KP}{N(\sigma_x^2 + K \sigma_v^2)} \mathbf{I}_N. 
\end{equation}

Hence, the optimal sensing matrix $\mathbf{A}$ (with respect to \eqref{eq:opt_rec_A_appx}), after rescaling to meet the power constraint, becomes \eqref{eq:closed_spec2_4}.

\section{Proof of \proref{prop:sepc_2}}
Following the assumption in \proref{prop:sepc_2}, the SDR optimization problem simplifies into
\begin{equation} \label{eq:spec2_ref_1}
\begin{aligned}
	&\underset{\mathbf{Q}  }{\text{minimize}} \hspace{0.25cm} \sum_\mathcal{S} \Tr \left\{\bigg(\frac{1}{\sigma_x^2} \mathbf{I}_K + \frac{g^2}{\sigma_w^2} \mathbf{E}_\mathcal{S}^\top \mathbf{H}^\top \mathbf{Q} \mathbf{H} \mathbf{E}_\mathcal{S}  \bigg)^{-1} \right\}& \\
	& \text{subject to} \hspace{0.25cm} \frac{\sigma_x^2}{K}  \; \Tr \{\mathbf{H}^\top  \mathbf{Q} \mathbf{H}\} \leq P .& 
\end{aligned}	
\end{equation}

The objective function in \eqref{eq:spec2_ref_1} reaches its minimum when $\mathbf{H}^\top \mathbf{Q} \mathbf{H} = \alpha \mathbf{I}_N$ (see \cite[Lemma 2]{03:Shengli}). Taking SVD, we have $\mathbf{H} = \mathbf{U}_H \mathbf{\Gamma}_H  \mathbf{V}_H^\top$, where $ \mathbf{U}_H$ and  $\mathbf{V}_H$ are $N \times N$ unitary matrices and $\mathbf{\Gamma}_H = \mathrm{diag}(\gamma_{h_1},\gamma_{h_2}, \ldots, \gamma_{h_N} )$ is a diagonal matrix containing singular values $\gamma_{h_1}<\gamma_{h_2}< \ldots < \gamma_{h_N}$. Then, it follows that the optimal $\mathbf{Q}$ should have the following structure
\begin{equation} \label{eq:spec2_ref_2}
	\mathbf{Q}^\star = \alpha (\mathbf{HH}^\top)^{-1} =  \alpha \mathbf{U}_H \mathbf{\Gamma}_H^{-2} \mathbf{U}_H^\top,
\end{equation}
where by plugging \label{eq:spec2_ref_2} into the power constraint, we obtain $\alpha = \frac{KP}{N \sigma_x^2}$. Therefore, the optimal sensing matrix $\mathbf{A}$ (with respect to \eqref{eq:opt_rec_A_appx}) can be chosen as in \eqref{prop:sepc_3}.

\section{Proof of \proref{prop:sepc_3}}
Having the assumptions in \proref{prop:sepc_3}, the oracle estimator in \eqref{eq:oracle MMSE est} can be written as
\begin{equation} \label{eq:oracle_est_spe_3}
\widehat{\mathbf{x}}^{(or)} \!=\! 
	g \left(\frac{g^2 \sigma_v^2}{\sigma_x^2} \mathbf{I}_K \!+\! g^2   \mathbf{A}_\mathcal{S}^\top \; (\mathbf{A}\mathbf{A}^\top)^{\! \dagger} \; \mathbf{A}_\mathcal{S} \right)^{\!-1} \!\!\mathbf{A}_\mathcal{S}^\top   (\mathbf{A}\mathbf{A}^\top)^{\! \dagger} \mathbf{y},
\end{equation}
where $(\cdot)^\dagger$ denotes matrix pseudo-inverse. Using $\mathbf{A}_\mathcal{S} = \mathbf{A} \mathbf{E}_\mathcal{S}$, it gives the oracle MSE
\begin{equation}
	\mathrm{MSE}^{(lb)} \!= \! \! \frac{1}{{N \choose K}} \! \sum_\mathcal{S} \Tr \! \left \{ \! \left(\frac{1}{\sigma_x^2} \mathbf{I}_K \!+\! \frac{1}{\sigma_v^2} \mathbf{E}_\mathcal{S}^\top \mathbf{A}^{\!\top} (\mathbf{A}\mathbf{A}^{\! \top})^{\! \dagger} \mathbf{AE}_\mathcal{S} \! \right)^{\!-1} \! \right \}.
\end{equation}

Taking SVD, $\mathbf{A} = \mathbf{U}_a [\mathbf{\Gamma}_a \:\: \mathbf{0}_{N-M}] \mathbf{V}_a^\top$, it follows that 
\begin{equation} \label{eq:svd A}
	\mathbf{A}^{\!\top} (\mathbf{A}\mathbf{A}^{\! \top})^{\! \dagger} \mathbf{A} = \mathbf{V}_a 
	\left[ \begin{array}{c c} 
	   \mathbf{I}_M  & \mathbf{0}_{M \times (N-M)} \\ 
	   \mathbf{0}_{(N-M) \times M}  &   \mathbf{0}_{(N-M)\times (N-M)}  \\
	\end{array} \right]
	\mathbf{V}_a^\top.
\end{equation}

Applying \eqref{eq:svd A} into \eqref{eq:oracle_est_spe_3}, the sensing matrix optimization problem can be posed as following
\begin{equation} \label{eq:spec3_opt_prob}
\begin{aligned}
	&\underset{\mathbf{\Gamma}_a, \mathbf{V}_a }{\text{minimize}} \hspace{0.25cm} \sum_\mathcal{S} \Tr \left\{\bigg(\frac{1}{\sigma_x^2} \mathbf{I}_K \!+\! \frac{1}{\sigma_v^2} \mathbf{E}_\mathcal{S}^\top \mathbf{V}_a \!
	\left[ \! \! \begin{array}{c c} 
	   \mathbf{I}_M  & \mathbf{0}\\ 
	   \mathbf{0} &   \mathbf{0}  \\
	\end{array} \! \! \right] \! \mathbf{V}_a^{\! \top}
	 \mathbf{E}_\mathcal{S}  \bigg)^{\!-1} \right\}& \\
	& \text{subject to} \hspace{0.25cm} \frac{\sigma_x^2}{K}  \; \Tr \{  \mathbf{\Gamma}_a^2 \} \leq P .& 
\end{aligned}	
\end{equation}

We note that the objective function in \eqref{eq:spec3_opt_prob} can be minimized with respect to $\mathbf{U}_a$ independent of $\mathbf{\Gamma}_a$ in the constraint. Now, since $\mathbf{E}_\mathcal{S}^\top \mathbf{V}_a \mathbf{V}_a^\top \mathbf{E}_\mathcal{S} = \mathbf{I}_K$, the objective function in \eqref{eq:spec3_opt_prob} can be rewritten and lower-bounded as 
\begin{equation} \label{eq:rewrite_lb}
\begin{aligned}
	&\sum_\mathcal{S} \Tr \left\{\bigg( \mathbf{E}_\mathcal{S}^\top \mathbf{V}_a \!
	\left[ \! \! \begin{array}{c c} 
	   (\frac{1}{\sigma_x^2} \!+\! \frac{1}{\sigma_v^2}) \mathbf{I}_M  & \mathbf{0}\\ 
	   \mathbf{0} &   \frac{1}{\sigma_x^2} \mathbf{I}_{N-M}  \\
	\end{array} \! \! \right] \! \mathbf{V}_a^{\! \top}
	 \mathbf{E}_\mathcal{S}  \bigg)^{\!-1} \right\} & \\
	 &\geq \sum_\mathcal{S} \Tr \left\{\bigg( \mathbf{E}_\mathcal{S}^\top \mathbf{V}_a \!
	\left[ \! \! \begin{array}{c c} 
	   (\frac{1}{\sigma_x^2} \!+\! \frac{1}{\sigma_v^2}) \mathbf{I}_M  & \mathbf{0}\\ 
	   \mathbf{0} &   \frac{1}{\sigma_x^2} \mathbf{I}_{N-M}  \\
	\end{array} \! \! \right] \! \mathbf{V}_a^{\! \top}
	 \mathbf{E}_\mathcal{S}  \bigg)_{ii}^{\!-1} \right\}  &
\end{aligned}
\end{equation}
where by $(\cdot)_{ii}$ we denote the diagonal elements of the corresponding matrix. The equality in \eqref{eq:rewrite_lb} is satisfied if and only if  the matrix inside the trace-inverse operator becomes diagonal, which yields $\mathbf{V}_a = \mathbf{I}_N$. Also, from the constraint in \eqref{eq:spec3_opt_prob}, it follows that $\mathbf{\Gamma}_a$ can be an arbitrary diagonal matrix satisfying the transmission power constraint. For simplicity, we set $\mathbf{\Gamma}_a = \sqrt{\frac{KP}{M \sigma_x^2}} \mathbf{I}_M$. Hence, the optimal sensing matrix has the structure in \eqref{eq:spec_case_3}.

\section{Proof of \proref{prop:sepc_4} }
We have
\begin{equation} \label{eq:taylor ser}
\begin{aligned}
&	\mathrm{MSE}^{(lb)} = \frac{1}{{N \choose K}} \sum_\mathcal{S} \Tr \left\{\left(\mathbf{R}^{-1} + \frac{g^2}{\sigma_w^2} \mathbf{D}_\mathcal{S}^\top \mathbf{Q} \mathbf{D}_\mathcal{S} \right)^{-1} \right\} & \\ 
	&\overset{(a)}{=} \! \frac{1}{{N \choose K}} \! \sum_\mathcal{S} \Tr  \! \left\{\mathbf{R} \!-\! \frac{g^2}{\sigma_w^2}  \mathbf{R} \mathbf{D}_\mathcal{S}^\top \mathbf{Q} \mathbf{D}_\mathcal{S}  \mathbf{R} \right\} 
	\!+\! \mathcal{O}(\|\frac{g^2}{\sigma_w^2}   \mathbf{D}_\mathcal{S}^\top \mathbf{Q} \mathbf{D}_\mathcal{S} \|_F^2),&
\end{aligned}
\end{equation}
where $(a)$ follows from Taylor series for the inverse term inside the trace operator in the first equation. Now, since $\frac{g^2}{\sigma_w^2} \rightarrow 0$, by neglecting the higher moments and using linear property of the trace operator, the original optimization problem in \eqref{eq:opt 1} can be asymptotically approximated as
\begin{equation} \label{eq:opt asymp}
\begin{aligned}
	&\underset{\mathbf{Q}}{\text{maximize}} \hspace{0.25cm} \sum_\mathcal{S} \Tr \left\{\mathbf{R} \mathbf{D}_\mathcal{S}^\top \mathbf{Q} \mathbf{D}_\mathcal{S}  \mathbf{R} \right\} & \\
	& \text{subject to} \hspace{0.25cm} \Tr \{\mathbf{HR}_x \mathbf{H}^\top \mathbf{Q}\} 
	\leq P &  \\
		&\hspace{1.7cm} \mathbf{Q} \succeq \mathbf{0} , \hspace{0.15cm} \mathrm{rank}(\mathbf{Q}) = M.&
\end{aligned}
\end{equation}

Note that the objective function in \eqref{eq:opt asymp} can be rewritten as $\Tr \left \{ \left[\sum_\mathcal{S}  \mathbf{D}_\mathcal{S} \mathbf{R}^2 \mathbf{D}_\mathcal{S}^\top \right]   \mathbf{Q}  \right\} $ due to linear property of the trace operator. Now, defining the full-rank symmetric positive definite matrix $\mathbf{T} \triangleq \sum_\mathcal{S}  \mathbf{D}_\mathcal{S} \mathbf{R}^2 \mathbf{D}_\mathcal{S}^\top$, and denoting $\mathbf{T}^{1/2}\mathbf{Q}\mathbf{T}^{1/2} \triangleq \mathbf{L}$, the optimization problem in \eqref{eq:opt asymp} can be rewritten as 
\begin{equation} \label{eq:opt asymp 2}
\begin{aligned}
	&\underset{\mathbf{L}}{\text{maximize}} \hspace{0.25cm} \Tr \left\{\mathbf{L} \right\} & \\
	& \text{subject to} \hspace{0.25cm} \Tr \{\mathbf{T}^{-1/2} \mathbf{HR}_x \mathbf{H}^\top \mathbf{T}^{-1/2} \mathbf{L}\} 
	\leq P &  \\
		&\hspace{1.7cm} \mathbf{L} \succeq \mathbf{0},  \hspace{0.15cm} \mathrm{rank}(\mathbf{L}) = M.&
\end{aligned}
\end{equation}

Let $\mathbf{Z} \triangleq \mathbf{T}^{-1/2} \mathbf{HR}_x \mathbf{H}^\top \mathbf{T}^{-1/2}$ have the EVD $\mathbf{Z} = \mathbf{U}_z \mathbf{\Gamma}_z \mathbf{U}_z^\top$. We also decompose $\mathbf{L}$ as $\mathbf{L} = \mathbf{U}_l \mathbf{\Gamma}_l \mathbf{U}_l^\top$, where $\mathbf{U}_z$  and $\mathbf{U}_l$ are unitary matrices, and $\mathbf{\Gamma}_z$ and $\mathbf{\Gamma}_l$ are diagonal matrices. Further, $\mathbf{\Gamma}_l$ contains at most $M$ non-zero diagonal elements. In order to solve \eqref{eq:opt asymp 2}, we drop the rank constraint, and relax \eqref{eq:opt asymp 2} using \lemref{lem:constr bound} as
\begin{equation} \label{eq:opt asymp 3}
\begin{aligned}
	&\underset{\{ \gamma_{l_i} \}_{i=1}^L}{\text{maximize}} \hspace{0.25cm} \sum_{i=1}^L \gamma_{l_i} & \\
	& \text{subject to} \hspace{0.25cm} \sum_{i=1}^L \gamma_{z_i} \gamma_{l_i}
	\leq P &  \\
		&\hspace{1.7cm} \gamma_{l_i} \geq 0 \; \; , \; \; 1 \leq i \leq L,& 
\end{aligned}
\end{equation}
where $\gamma_{l_1} \geq \ldots \geq \gamma_{l_L}$ and $\gamma_{z_1} \leq \ldots \leq \gamma_{z_L}$.

Note that the optimization problem \eqref{eq:opt asymp 2}, without the rank constraint, and \eqref{eq:opt asymp 3} become equivalent when $\mathbf{Z L}$ is diagonal. This holds when $\mathbf{U}_l = \mathbf{U}_z$, where the columns of $\mathbf{U}_z$ are associated with the eigen-values of $\mathbf{Z}$ in an increasing order. Now, it only remains to solve \eqref{eq:opt asymp 3} for the singular-values $\gamma_{i}$'s ($1 \leq i \leq M$). It is well-known that the objective function in \eqref{eq:opt asymp 3} is maximized by letting $\gamma_{l_1} = \frac{P}{\gamma{z_1}}$, and $\gamma_{l_2} = \ldots = \gamma_{l_L} = 0$. Thus, it follows that
\begin{equation} \label{eq:sol asymp}
	\mathbf{Q}^\star =  \mathbf{T}^{-1/2}\mathbf{U}_z \mathrm{diag} \left(\frac{P}{\gamma_{z_1}},0,\ldots,0 \right)\mathbf{U}_z^\top \mathbf{T}^{-1/2}.
\end{equation}

Note that from \eqref{eq:sol asymp}, it is observed that the low rank matrix $\mathbf{Q}^\star$ in \eqref{eq:sol asymp} has only one non-zero eigen-value. Hence, using EVD of $\mathbf{Q}^\star$, we have $\mathbf{Q}^\star = \mathbf{U}_q  \mathrm{diag} \left(\gamma_q,0,\ldots,0 \right)  \mathbf{U}_q^\top$, where $\gamma_q > 0$ denotes the non-zero eigen-value of $\mathbf{Q}^\star$. Now, let the SVD of $\mathbf{A}$ be $\mathbf{A} = \mathbf{U}_a [\mathbf{\Gamma}_a \; \; \mathbf{0}_{M \times (L-M)}]\mathbf{V}_a^\top$, where $\mathbf{U}_a \in \mathbb{R}^{M \times M}$ and $\mathbf{V}_a \in \mathbb{R}^{L \times L}$ are unitary matrices, and $\mathbf{\Gamma}_a \in \mathbb{R}^{M \times M}$ is a diagonal matrix. From $\mathbf{Q} = \mathbf{A}^\top \mathbf{A}$, it is concluded that the optimal sensing matrix can be expressed as
\begin{equation} \label{eq:sol asymp2_prrof}
	\mathbf{A}^\star =  \mathbf{U}_a \left[\mathrm{diag}  \left(\sqrt{\gamma_q},0,\ldots,0 \right)\; \; \mathbf{0}_{M \times (L-M)} \right] \mathbf{U}_q^\top.
\end{equation}

\end{appendices}

\bibliographystyle{IEEEtran}
\bibliography{IEEEfull,bibliokthPasha}

\end{document}